\documentclass{article}

\usepackage{fullpage}
\usepackage{epsfig}
\usepackage{amsmath}
\usepackage{amsthm}
\usepackage{amssymb}
\usepackage[usenames,dvipsnames]{xcolor}
\usepackage{ifpdf,cite,url}
\usepackage[draft]{hyperref}

\newtheorem{lemma}{Lemma}
\newtheorem{proposition}{Proposition}
\newtheorem{theorem}{Theorem}

\DeclareMathOperator*{\argmin}{arg\,min}

\DeclareMathOperator*{\img}{img}
\newcommand\inv{{-1}}

\begin{document}

\title{Computing minimum area homologies}


\author{
 Erin Wolf Chambers\thanks{Dept. of Math and Computer Science, Saint Louis University, {\tt echambe5@slu.edu}.  Research supported in part by the NSF under Grant No. CCF-1054779 and IIS-1319573.}
\and
Mikael Vejdemo-Johansson\thanks{Institute for Mathematics and its
  Applications, University of Minnesota, Minneapolis. Computer Vision
  and Active Perception Laboratory, KTH Royal
  Institute of Technology, Stockholm. AI Laboratory, Jozef Stefan
  Institute, Ljubljana. \texttt{mvj@kth.se}. Research supported in
  part by the FP7 project \textsc{Toposys} (FP7-ICT-318493-STREP).}
}
\date{}

\maketitle

\small
\begin{abstract}
Calculating and categorizing the similarity of curves is a fundamental problem which has generated much recent interest.  However, to date there are no implementations of these algorithms for curves on surfaces with provable guarantees on the quality of the measure.  In this paper, we present a similarity measure for any two cycles that are homologous, where we calculate the minimum area of any homology (or connected bounding chain) between the two cycles.  The minimum area homology exists for broader classes of cycles than previous measures which are based on homotopy.  It is also much easier to compute than previously defined measures, yielding an efficient implementation that is based on linear algebra tools.  We demonstrate our algorithm on a range of inputs, showing examples which highlight the feasibility of this similarity measure.
\end{abstract}

\section{Introduction}
\label{sec:introduction}

Recently, much work has been done in the area of similarity measures for curves in topological spaces.  For example, in the graphics, vision, and medical imaging communities, the goal is often to measure  similarity between objects~\cite{517122,1314502}, with the goal of categorizing or recognizing similar structures.  Questions regarding similarity between maps or trajectories within a map are motivated by the vast quantity of data from GIS systems; see for example~\cite{mapsurvey} for a recent survey of some algorithms in this area.

In the computational geometry community, much of the focus has been on the Fr\'echet distance in Euclidean space~\cite{ag-cfdbt-95,alt2009} and its generalizations such as geodesic and homotopic Fr\'echet distance~\cite{CW10,CVE08}.  However, none of these have proven algorithmically tractable to compute on surfaces.  There is also work giving an algorithm to compute the minimum area homotopy between two curves or cycles on surfaces; somewhat surprisingly, computing area is much more practical, giving algorithms which are either quadratic or linear in time (depending on the exact setting)~\cite{homotopyarea}.  Although this minimum homotopy algorithm has not been implemented, it is the first theoretically tractable algorithm in the surface setting.

The use of homology rather than homotopy was suggested in prior work
as a natural measure to consider~\cite{homotopyarea}; this is perhaps
especially appealing given that homology computations often have  extremely
efficient implementations which are based on operations over vector
spaces.  For example, there are efficient algorithms which compute the
minimum cycle homologous to an input cycle~\cite{dhk2011}.  In
addition, there is a large body of work on calculating homology groups
and generators in a variety of settings. In addition to being a
standard exercise in linear algebra, tools for computing simplicial
homology -- either of triangulations or induced from point cloud data
-- are commonplace, and occur in software packages including GAP~\cite{gap,hap,pgap},
Magma~\cite{magma}, and Sage~\cite{sage}, as well as more specialized libraries such as ChomP~\cite{chomp},
javaPlex~\cite{javaplex}, Dionysus~\cite{dionysus},
Perseus~\cite{perseus}, pHat~\cite{phat} and Kenzo~\cite{kenzo}.

However, the concept of finding a minimum area homology has not yet been investigated from the algorithmic  perspective.  This is perhaps due to the fact that  homotopy describes a more natural notion of a continuous deformation, while with homology it is much less clear what geometric structure is described.  There is a clear connection between the two; indeed, for higher dimensional manifolds, the problem of finding a minimum homotopy reduces to that of finding a minimum homology~\cite{white1984}.  Moreover, for the case of cycles on a surface, two homotopic cycles clearly have a homology between them, and actions such as puncturing the surface or introducing high weight simplices will often force the minimum area homology to agree with the minimum homotopy.

In this paper, we formalize a notion of a homology between two cycles with minimum area, and prove several interesting properties.  We also implement algorithms to compute this area, concluding with a series of examples  to demonstrate the success of this measure.

It is also worth noting that by working with homology instead of homotopy, we can construct bounding
chains for a wider class of structures on a surface. Recall that two
(not necessarily connected) dimension $k$ submanifolds $S, T$ of a manifold
$M$ are said to be \emph{cobordant} if there is a dimension $k+1$
submanifold $W$ of $M$ with boundary $S\sqcup T$. The connecting
submanifold $W$ is said to be the \emph{cobordism} from $S$ to $T$.
Cobordisms are one of the most important techniques in contemporary
algebraic and differential topology, and were fundamentally important
in some of the most famous theorems in the field, including the
Hirzebruch-Riemann-Roch theorem~\cite{hrr} and the Atiyah-Singer index
theorem~\cite{asit}. A fundamental survey of cobordisms and their uses
was written by Milnor in 1962~\cite{milnor}.

For two arbitrary closed collections $S, T$ of disjoint simple curves on a surface, computing a
connecting chain generates a cobordism from $S$ to $T$.  So while we we focus solely on minimum area homologies in this work, our algorithm to compute a connecting chain generalizes to significantly more complex structures in higher dimensions.

\section{Definitions}
\label{sec:definitions}

In this paper we make extensive use of simplicial homology for triangulated
surfaces, and its properties.  We provide a brief overview of this area, and refer to the reader to a standard reference on topology for full details~\cite{h-at-02,m-t-00}.

Recall that given a triangulation of a
surface $M$ with
vertices $V$, edges $E$ and faces $F$, there are \emph{chain vector
  spaces} $\mathbb{R}^V, \mathbb{R}^E, \mathbb{R}^F$. Vectors in these
vector spaces correspond to weighted collections of faces, edges,
vertices. There is a linear \emph{boundary operator} $\partial_2:
\mathbb{R}^F\to \mathbb{R}^E$ and
$\partial_1:\mathbb{R}^E\to\mathbb{R}^V$. Closed curves composed out of the triangulation edges are
represented by vectors in $\ker\partial_1$, and curves that occur as
boundaries of collections of triangles are closed curves and are in $\img\partial_2$, so that $\img \partial_2 \subseteq \ker\partial_1$. The
vector space of \emph{essential closed curves} is represented by the
homology group $H_1(M; \mathbb{R})=\ker\partial_1/\img\partial_2$. In
particular, we call two vectors of edge combinations $z, w$
\emph{homologous} if $z-w\in\img\partial_2$; in this case, $z$ and
$w$ represent the same essential closed curve.

We can now consider a \emph{minimum area homology} between $z$ and $w$
- namely, a connected bounding chain of triangles $c$ where $\partial c = z-w$ and
the area of $c$ is the infimum over all possible such bounding chains. Here, and later in the paper, we use the term \emph{bounding chain} to denote precisely some $c$ such that $\partial c=z-w$. For instance, a simple closed curve on the surface not encycling any tunnels or holes would have the interior disk as its bounding chain.
This minimum area measure between curves is naturally commutative and
satisfies a notion of the triangle inequality, since homology is an
equivalence relation, and therefore transitive.

Homotopy, the more usual measure of similarity in this area, is perhaps more intuitive, since a \emph{homotopy} between two cycles is simply a continuous deformation between them.  Homology is a coarser measure than homotopy, since the homology group is a abeliazation of the homotopy group; any homotopic cycles are necessarily homologous, but the reverse is not true.  However, when calculating area, it is not clear that minimum area homotopies even exist, and great care is taken in the mathematical literature to restrict the type of integral in order to guarantee existence; see the book by Lawson for an overview of this problem and its historical developments~\cite{lawson1980lectures}.

In contrast, minimum area homologies automatically exist (so the infimum 
becomes a minimum) and are 
computable given the fact that we are working in a vector space.  In
addition,  homology has the added advantage that we can consider
more than simple cycles which are continuously deformable to each
other, since homology allows one to deform ``over" handles and other
topological features.   See Figure~\ref{fig:homotopyversushomologyonsurface} for an example of homologous cycles on a surface which are not homotopic, as well as Figures~\ref{fig:chair} and \ref{fig:crab}.

\begin{figure}
  \centering
  \includegraphics[width=3in]{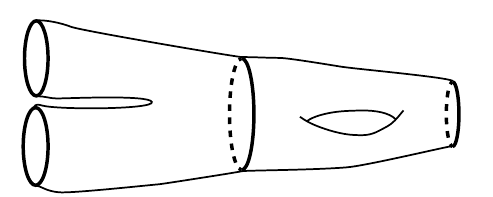}
  \caption{Two examples of homologous cycles on a surface that are not
    homotopic. The middle cycle is homologous both to the right cycle
    and to the sum of the left two cycles; the rightwards homology is
    not a homotopy due to the hole in the surface, the leftwards one
    due to the split of one cycle into two. We will see examples of
    both behaviours later in the paper, in Figure~\ref{fig:chair}.}
  \label{fig:homotopyversushomologyonsurface}
\end{figure}


\section{Connection to homotopy}
\label{sec:minim-area-homol}

In this section, we first focus on the case of finding the minimum area homology for two homotopic input cycles, since we expect this to be the most interesting application of our algorithm.  The following lemmas establish existence of minimum area homologies in this case and the connection between the minimum area homology and minimum area homotopy.

\begin{lemma}
\label{lem:homotopy}
Suppose we have two cycles $z,w: S^1 \rightarrow M$ on a manifold where $z$ is homotopic to $w$.  Then there exists a connected bounding chain $c$ such that $\partial c = z-w$.
\end{lemma}

\begin{proof}
Since $z$ and $w$ are homotopic, there is a homotopy $H: S^1 \times I \rightarrow M$ where $H(\cdot,0) = z$ and $H(\cdot,1) = w$.  We know that $S^1 \times I$ is connected, so let $c = \img(H)$.
\end{proof}

Note that if $M$ is triangulated, both cycles consist of edges in the triangulation, and there is a homotopy built from the triangulation of $M$, then this lemma holds as is for the simplicial case.

The homology equivalence class of connected bounding chains in
Lemma~\ref{lem:homotopy} is a finite dimensional affine linear space
if $M$ has a finite triangulation. The area measure is a continuous
function to the non-negative reals, and thus attains its minimum.  Therefore,
we can refer to the minimum area homology as the bounding chain which attains the infimum
in this area measure.

For well-behaved cases in a 2-sphere, we can in fact narrow the candidates down more concretely.

\begin{proposition}
\label{prop:2sphere}
Consider two disjoint cycles $z,w: S^1 \rightarrow S^2$ on a sphere, with $z$ homotopic to $w$.  Then there are only 2 supports for a connected bounding chain, one of which describes the minimum area homotopy.
\end{proposition}

\begin{proof}
Consider two cycles $z$ and $w$ that are disjoint.  Fix $z$, and the Jordan curve theorem immediately implies that there are two disjoint sides of the $S^2 - z$.  Since $w$ is disjoint, it is on one side or the other, so the bounding chain is either the region of $S^2 - z - w$ that is "between" $z$ and $w$, or is the bounding chain which is the union of the disk bounded by $z$ and the disk bounded by $w$.
\end{proof}

As with Lemma~\ref{lem:homotopy}, the result here holds as is if the cycles respect an existing triangulation of the sphere.

\begin{figure}
\begin{center}
\includegraphics[width=3in]{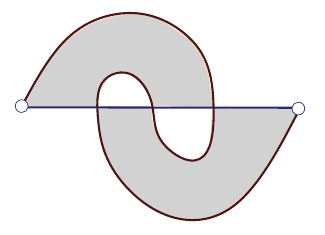}
\includegraphics[width=3in]{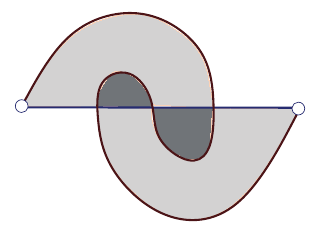}
\caption{Top: The minimum area homology between these two curves consists of the sum of the shaded areas.  Bottom: The minimum area homotopy, in contrast, must sweep regions "twice"; since homotopy does not cancel, the area of the small darkly shaded regions is counted twice in the area of the minimum homotopy, giving a strictly greater value. We revisit this type of example on a scanned mesh in Figure~\ref{fig:crab}.}
\end{center}
\label{fig:homotopy-homology}
\end{figure}

In the case that cycles intersect, the situation is more complex, since there are more than 2 connected bounding chains for such inputs.  See~\cite{homotopyarea} for a more thorough discussion of this complexity in the context of homotopy; since each homotopy class also gives a valid homology between the curves, the same issue exists in our setting.  

Since the image of a homotopy produces a connected bounding chain, we can also observe that the minimum area homotopy is bounded below in area by the minimum area homology.  However, there are cases (even for homotopy between simple curves in the plane) where the minimum area homotopy is strictly greater than the minimum area homology; see Figure~2.

\section{Algorithm}
\label{sec:algorithm}

Our approach to an algorithm and an implementation for the minimum
homology area is rooted in the extensive work available on finding
small solutions to linear systems of equations.

\subsection{Optimization problem formulation}
\label{sec:optim-probl-form}

Computing the minimum area bounding chain for a cycle $z \in C_1M$ (where here we use $z$ 
to denote the difference of the two input cycles rather than one of the two inputs)
is a question of computing a minimum area solution to a linear set of
equations. A bounding chain for $z$ is some $x$ such that $dx =
z$. Hence, we can formulate the minimum area bounding chain by solving
the optimization problem
\begin{equation}
  \label{eq:areaprob}
  \argmin_x\sum_j a(\sigma_j)|x_j| \qquad \text{subject to } \partial x=z
\end{equation}
where $a(\sigma_j)$ is the area contribution from cell
$\sigma_j$.

There are several variations of this optimization problem that turn
out to be useful for implementation, speed or practicality.

If we write $A$ for the diagonal matrix with entries $a(\sigma_j)$, we
note that $\sum_j a(\sigma_j)|x_j| = \|Ax\|_1$. We can recast 
Problem \ref{eq:areaprob} as a pure $L_1$-minimization problem. To do
this, we note that if $x'=Ax$ and all areas are non-zero, then $\partial x = \partial A^\inv x'$. Thus, $x$
minimizing $\|Ax\|_1$ subject to $\partial x=z$ can be computed by first
computing $x'$ minimizing $\|x'\|_1$ subject to $\partial A^\inv x' = z$. Thus,
we can also solve for the minimum area bounding chain by solving the
optimization problem
\begin{equation}
  \label{eq:l1prob}
  \argmin_x\|x\|_1 \qquad \text{subject to }\partial A^\inv x = z
\end{equation}

This equation captures a known optimization problem in machine
learning and compressed sensing\cite{donoho2006most}, which can be solved i.a. using
the Orthogonal Matching Pursuit (OMP)
algorithm\cite{mallat1993matching,rubinstein2008efficient}. 

It is worth noting for the
specific case of surface meshes in $\mathbb{R}^3$ that have no
boundaries, the results by \cite{dhk2011} offer up a linear
programming approach to the optimization problem. Running in $O(n^3)$,
this is an attractive option for an exact solution. However, this does not cover all the cases we consider in this paper, and restricts the
generalizability to higher dimensions. As pointed out in
\cite{dhk2011}, the presence of a M\"obius strip is sufficient to break
the total unimodularity assumption that the linear programming
approach relies on. (We would like to thank one of
our anonymous reviewers for alerting us to the linear programming approach.)

\subsection{Trade-off between $L_1$ and $L_2$}
\label{sec:trade-between-l_1}

While the computation of the $L_1$-minimization problem in
Section~\ref{sec:optim-probl-form} is a clear way to produce
an area-minimal connecting chain, this class of minimization problems
lacks closed form analytic solutions. In this section, we shall
describe a related $L_2$-minimization problem, as well as bounds that relate
it to the area-minimization problem.

It is an immediate consequence of the H\"older inequality and the
Cauchy-Schwarz inequality that $\|x\|_2 \leq \|x\|_1 \leq \sqrt{n}\|x\|_2$ for
any $x\in\mathbb{R}^n$. It follows that we can bound the area norm for
a finite triangulation with a weighted $L_2$-norm.

One way to get an approximate solution to the minimum area is to
compute $\argmin_{\hat x} \hat{x}^T \hat{x}$ such that $\partial\hat{x} =
\hat{z}$, where $\hat{x} = A^{1/2}x$.  Then $\hat{x}^T x = x^T
(A^{1/2})^T A^{1/2} x = x^T A x =\sum_j a(\sigma_j)x_j^2$.  This lets us solve $\argmin
\hat{x}^T\hat x$ such that $\partial \hat{x} = \hat{z}$, and then the appropriate $x$ is $A^{1/2}\hat{x}$.

Replacing the $L_1$ problems with $L_2$ problems in this way produces
these two formulations of area minimization as $L_2$ minimization problems:

\begin{align}
  \label{eq:l2wprob}
  \argmin_x x^T A x &\qquad \text{subject to }\partial x = z 
  \\
  \label{eq:l2prob}
  \argmin_x\|x\|_2 &\qquad \text{subject to }\partial A^{-1/2}x = z
\end{align}

A large benefit of this reformulation is that the $L_2$-problem
\ref{eq:l2prob} is the \emph{least squares} optimization problem and
has an analytic solution. Writing $W=\partial A^{-1/2}$, the vector
$z=(W^TW)^\inv W^Tw$ is a solution of minimum norm to $Wz=w$. Using
the bounds we mentioned above, this $L_2$-minimizer is an
approximation to the actual area minimal solution, and has the benefit
of being computable in matrix multiplication time (generally written as $O(n^\omega)$) with a
small constant for contemporary implementations, where $n$ is the
larger of the counts of triangles $T$ and edges $E$ in the
triangulation. OMP, as a comparison, can be implemented with
parametric complexity $O(n^\omega
k+(E+T)k+k^3)$ where $k$ is $\|z\|_1$.

\subsection{Bounded area variation}
\label{sec:bound-area-vari}

It is worth noting that our initial implementations, which neglected to include any weight calculations but instead relied only on the number of simplices in the homology, worked quite well on many of the initial test cases.  This is perhaps not terribly surprising, given that our initial inputs were well formed with relatively even sized triangles, so that comparing the number of triangles included a connected bounding chain formed a reasonable approximation to area.   Given the simplicity of the this code and speed-up of computation, we formalize this idea as follows.

\begin{theorem}\label{thm:boundedarea}
  Suppose the area of all triangles $\sigma$ in a triangulation of a
  surface $M$ are bounded by $a\leq A(\sigma) \leq A$.

  If $z$ minimizes $\|z\|_1$ subject to $\partial z = w$ and $z'$ minimizes
  $A(z')$ subject to $\partial z' = w$, so that $z$ is $L_1$-minimal and
  $z'$ is area minimal, then $a\|z\|_1 \leq A(z') \leq A\|z\|_1$. 
\end{theorem}

\begin{proof}
  Since $z$ is $L_1$-minimal, we know that $\|z\|_1 \leq
  \|z'\|_1$. Since $a$ is a lower bound of the triangle areas, this
  establishes the lower bound.

  Since $z'$ is area minimal and $A$ is an upper bound on triangle
  areas, it follows that $A(z') \leq A(z) \leq A\|z\|_1$. This
  establishes the upper bound.
\end{proof}

This means that even without solving for an area minimal bounding
chain we are able to establish a confidence interval for the area
minimal bounding chain; and unless the local triangulation density
varies considerably, we are likely to find the area minimal bounding chain by
simply minimizing the triangle count.


\subsection{Higher dimensions}
\label{sec:higher-dimensions}

Our implementation in the next section focuses on triangulated meshes, simply because that is the data most readily available as well as the most immediate application.
However, the algorithm we describe works for arbitrary dimensions of meshes and
embedded substructures: curves or surfaces in volume meshes or higher
dimensional analogues. For a $k$-dimensional submesh of a
$d$-dimensional mesh, the area bounds in Theorem~\ref{thm:boundedarea}
as well as the entire analysis of the optimization problem formulation
and its complexity carry over with the bounding chain taken to be a
$k+1$-dimensional submesh. In this setting, ``area'' would be replaced
by its $k+1$-dimensional correspondence.

\section{Results}
\label{sec:results}

In our implementation, we store the areas (or weights) as a diagonal matrix $A$, and let $D$ be the matrix which stores the boundary operator $\partial$.  Our algorithm should compute $\argmin x^T A x$ such that $Dx = z$, where $z$ is the difference of the two input cycles whose minimum area homology we wish to compute.  In the following sections, we describe the various implementations and approximations we pursued, along with results from each.

\subsection{Proof of concept implementation}
\label{sec:proof-conc-impl}

We can modify our problem to look like a standard machine learning algorithm  as follows.  In orthogonal matching pursuit, the goal is to compute $\argmin \|\hat{x}\|_1$ such that $D\hat{x} = \hat{z}$~\cite{342465}, and code to compute has been done for many languages, including the scikit-learn package in python~\cite{scikit-learn}.  Here, we will let $\hat{x} = Wx$, so that $Dx = DW^{-1}W x = DW^{-1} \hat{x} = z$, so that our new problem is $\argmin_{\hat{x}} |\hat{x}|_1$ such that $DW^{-1} \hat{x} = z$, and we can then let $x = W^{-1}\hat{x}$.  This lets us use the scikit package to solve our algorithm.

We also have code that computes $\argmin \|\hat x\|_2$ subject to
$D\hat x=\hat z$, using SciPy for computations~\cite{scipy}. All our
code and our experiments are made available at
\url{http://bitbucket.org/michiexile/minimumhomology}. 

\subsection{Sample outputs}
\label{sec:pretty-pictures}

Our images were primarily pulled from the PLY models curated by John
Burkardt~\cite{burkardt} and the Smithsonian X 3D
repository~\cite{3dsi-chair,3dsi-mammoth}.

Our own code was implemented  in Python, using SciPy~\cite{scipy} for
linear algebra and Shapely~\cite{shapely} for polygonal area
computations.  In general, we were able to run it quickly on meshes
containing up to 10\,000 triangles, which for the Smithsonian meshes
required pre-processing with MeshLab~\cite{Meshlab} to simplify and find candidate cycles on the meshes.

In Figures~\ref{fig:bigfig} and \ref{fig:secondbigfig}, 
minimum homologies are shown which were computed using simply a count of the number of triangles included in the homology.  Note that as the relevant regions of these meshes have relatively uniform area bounds on their triangles, this preliminary version of the algorithm works quite well; see Section~\ref{sec:bound-area-vari} for the relevant proofs of why these bounds hold.  

\begin{figure}
\begin{tabular}{ccc}
\includegraphics[width=.33\linewidth]{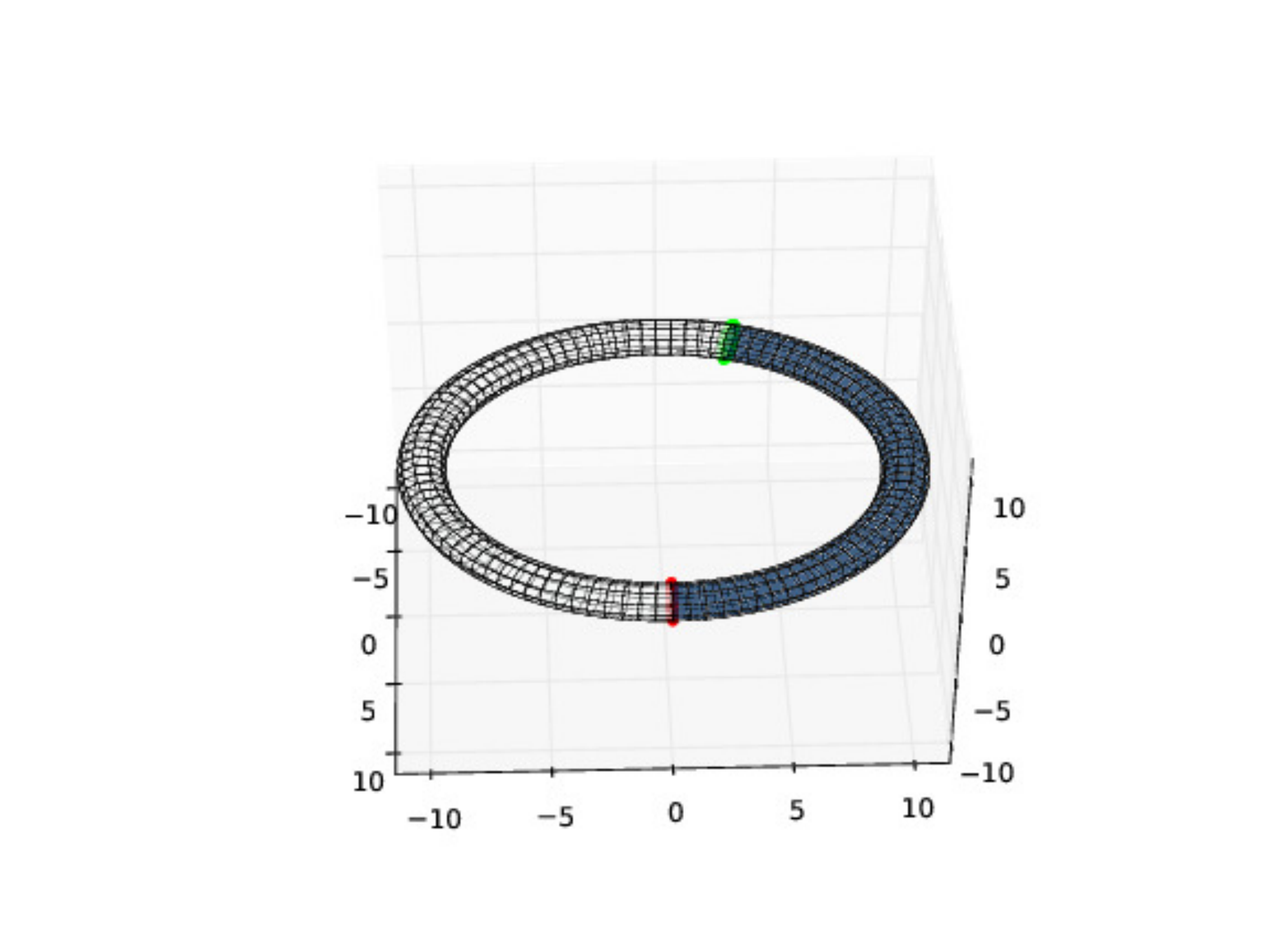}
\includegraphics[width=.33\linewidth]{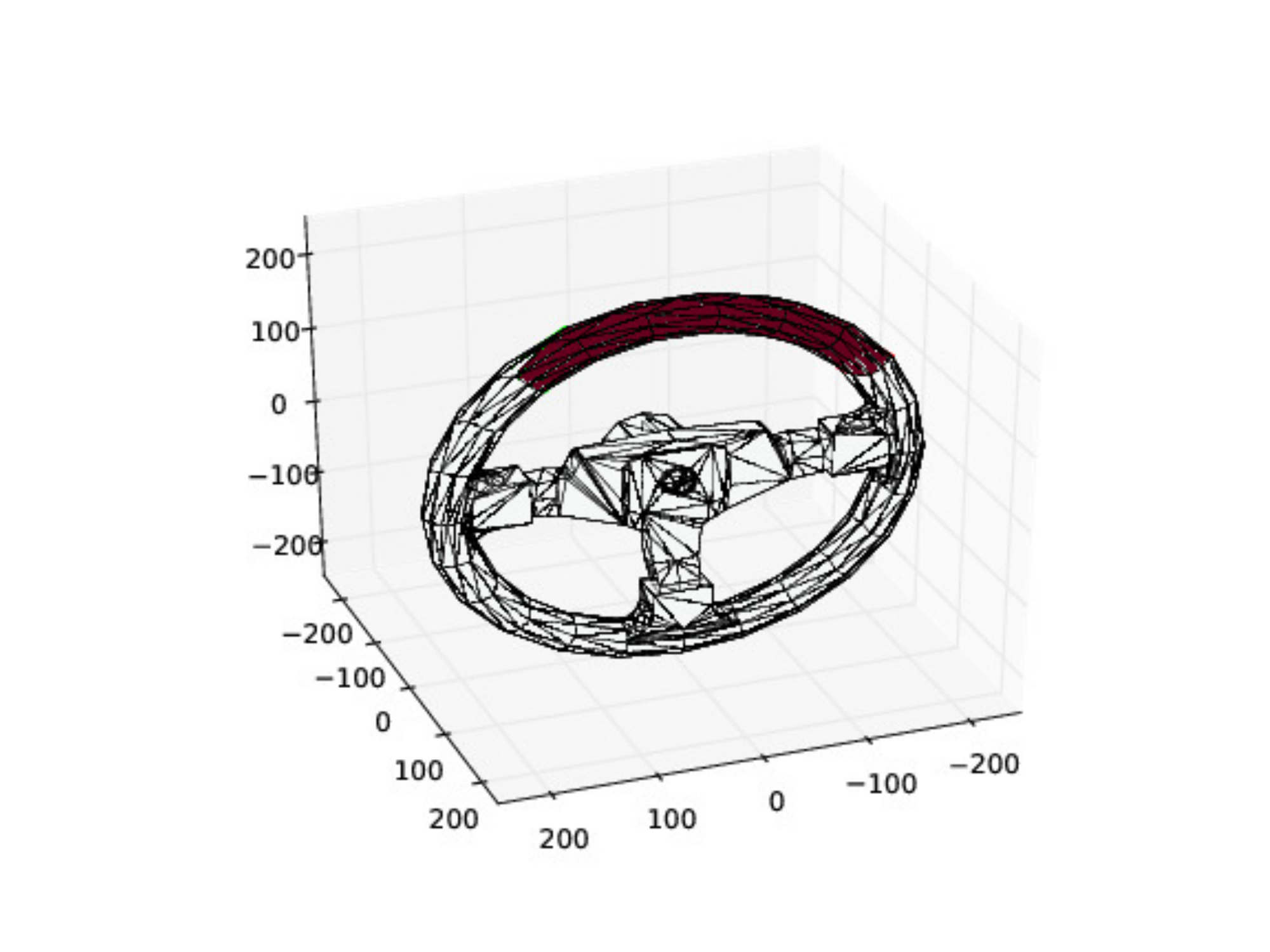}
\includegraphics[width=.33\linewidth]{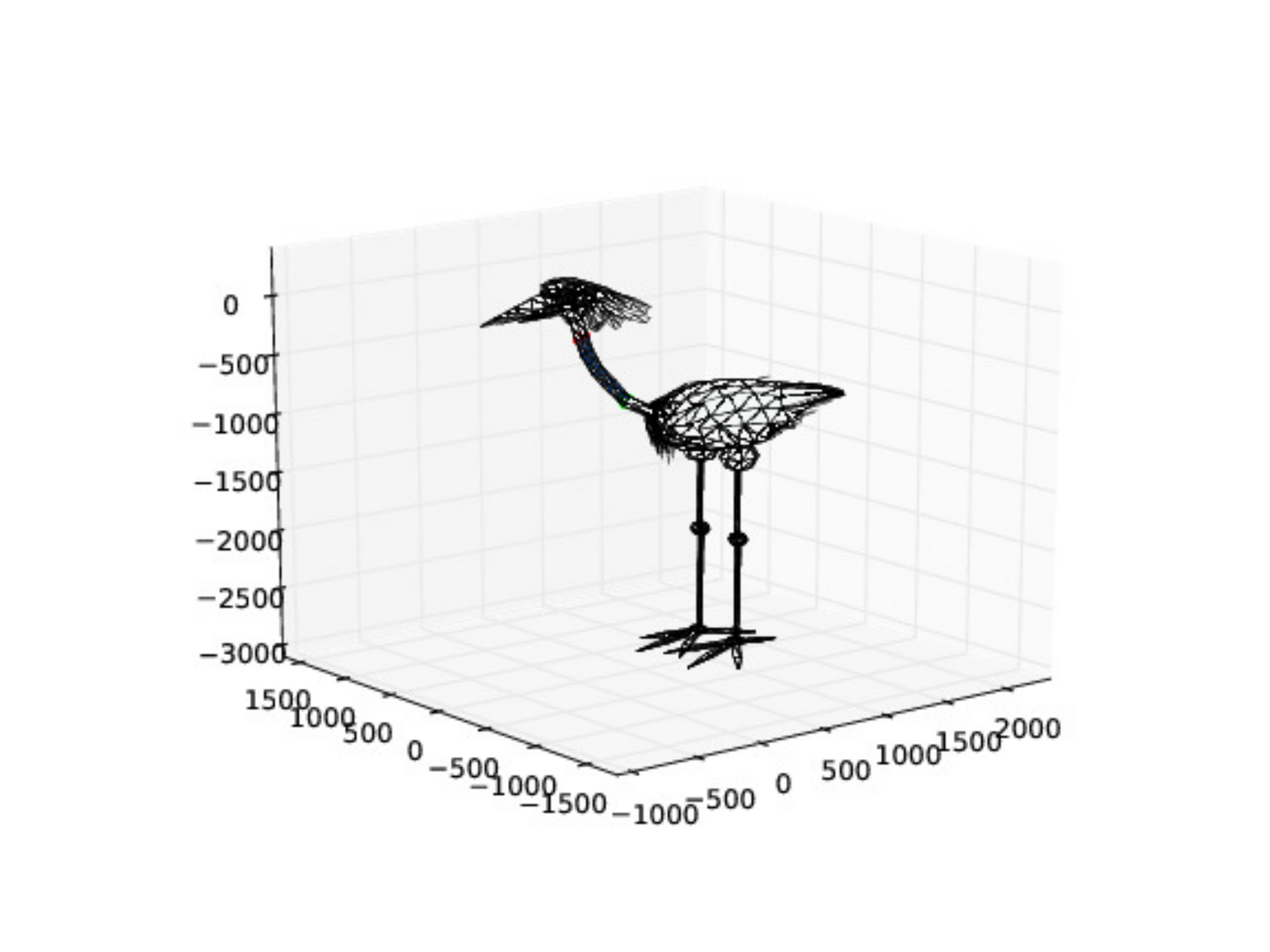}
\end{tabular}
\caption{The minimum homology (in terms of the number of triangles) between two cycles is shown shaded in each example above.}
\label{fig:bigfig}
\end{figure}

\begin{figure}
\begin{tabular}{ccc}
\includegraphics[width=.45\linewidth]{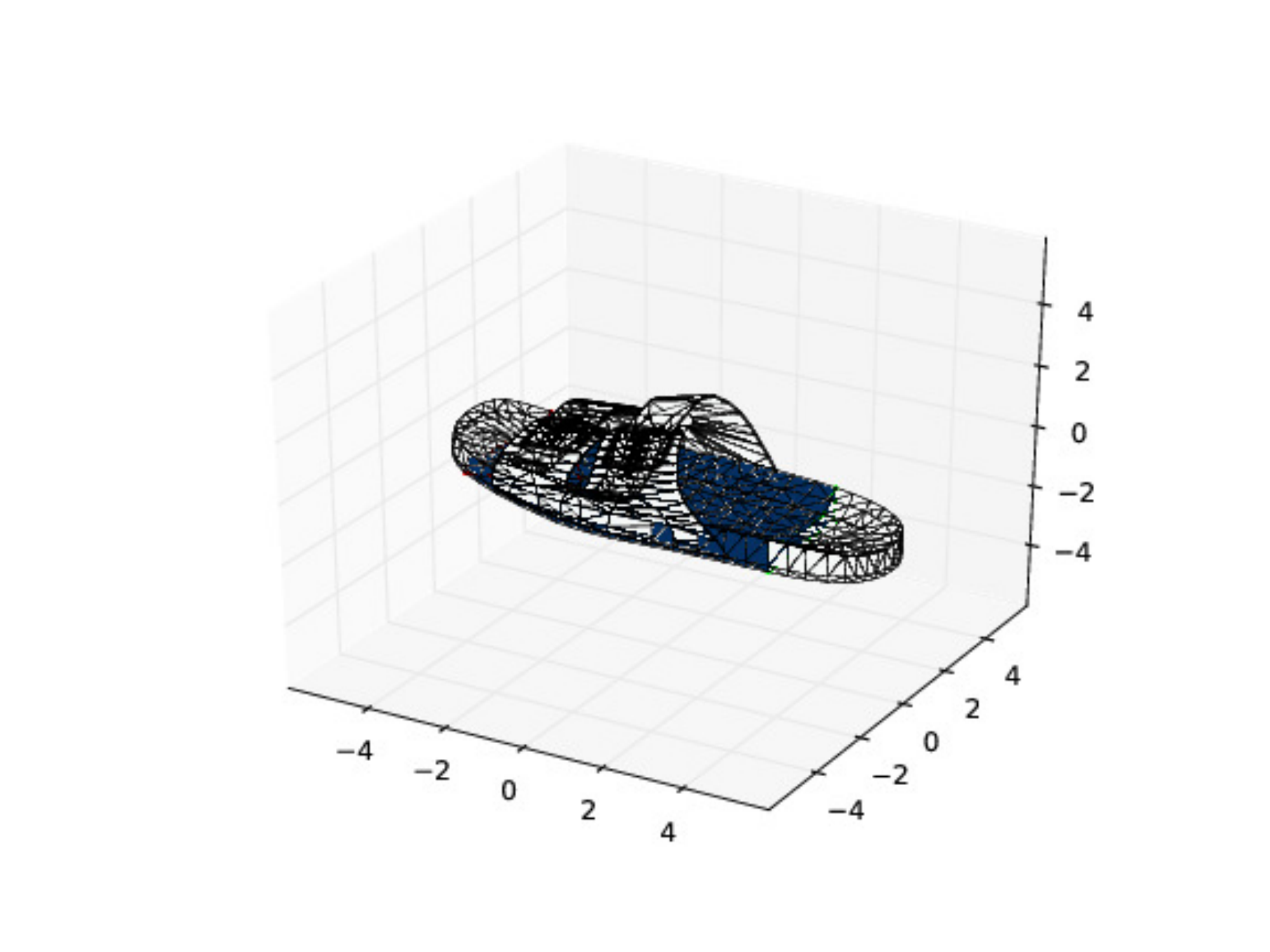}
\includegraphics[width=.45\linewidth]{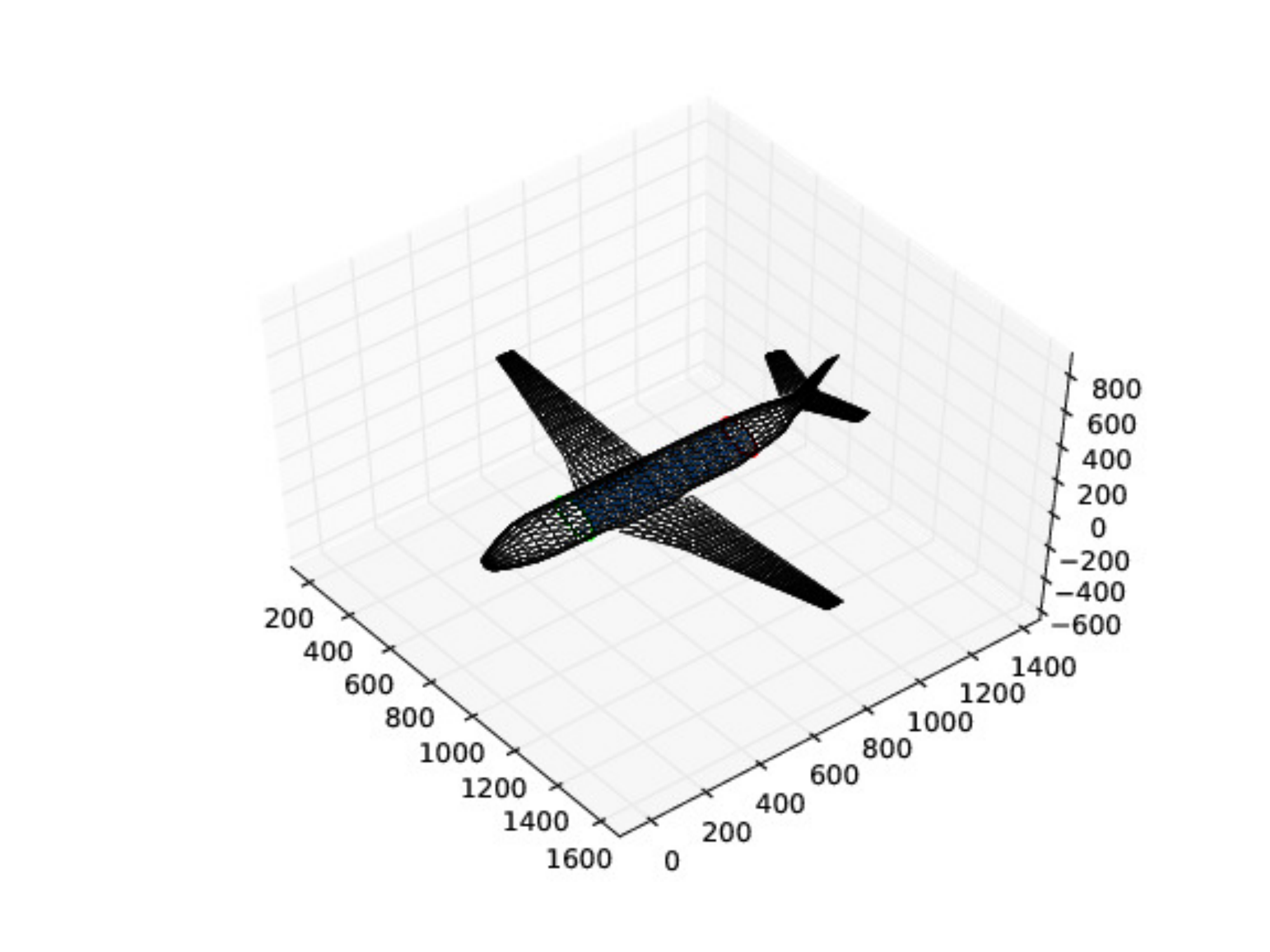}
\end{tabular}
\caption{Left: The minimum area homology (shown shaded).  Note that because the mesh does not connect the top and bottom of the sandal, the top portion is not included in the area calculation. Right: Again, note the issues in the mesh quality; since the wings are not actually connected to the body of the airplane, the wings are not included in the calculation, although ideally they should be since they would greatly increase the area between these two cycles.}
\label{fig:secondbigfig}
\end{figure}

%

%
%
%
%

Note that in Figure~\ref{fig:secondbigfig}, an interesting artifact of the mesh quality was inadvertently revealed in both examples.  A minimum homology between cycles according to the mesh is shown; however, in each of these examples, the mesh itself consists of disconnected pieces. The airplane actually has the wings detached in the mesh representations, and the sandal's top is a separate piece from the sole.  Therefore, our minimum area homology discounts these pieces.

In the Smithsonian meshes, results are far better. One favorable feature is that
the meshes are curated and post-processed for 3d-printing, and
therefore avoid the kind of self-intersection singularities we see in
Figure~\ref{fig:secondbigfig}. Another feature
is that these meshes are given in a coordinate system with millimeter
units, so area computations reflect actual areas on the meshed physical
artifacts. 

In Figure~\ref{fig:mammoth}, we show several different minimum area
homologies between pairs of input cycles; the area swept by the
homology is shown shaded, along with the distances between them given
in meters squared. We picked two cycles around the left tusk, and two
cycles around the spine, and compute all pairwise distances between
these four cycles.  

In Figure~\ref{fig:chair}, we can see even more
clearly the advantages of using homology, since the indicated cycles
are all homologous but not homotopic. We compute all pairwise
distances for three different cycles: two single closed loops and a
homologous pair of closed loops, mimicking the structure in
Figure~\ref{fig:homotopyversushomologyonsurface}. Again, areas between
them are shown shaded, and the distances here are given in thousands
of millimeters squared.  

Finally, in Figure~\ref{fig:crab}, we demonstrate more minimum homologies between cycles (or collections of cycles, in the case of the two cycles at the base of the claws).  We note in this image the results when cycles intersect, such as in the bottom left; this demonstrates a more concrete example of the simple curves shown in Figure~\ref{fig:homotopy-homology}.

\begin{figure*}
  \centering
  \begin{tabular}{ccc}
    \includegraphics[width=0.475\linewidth]{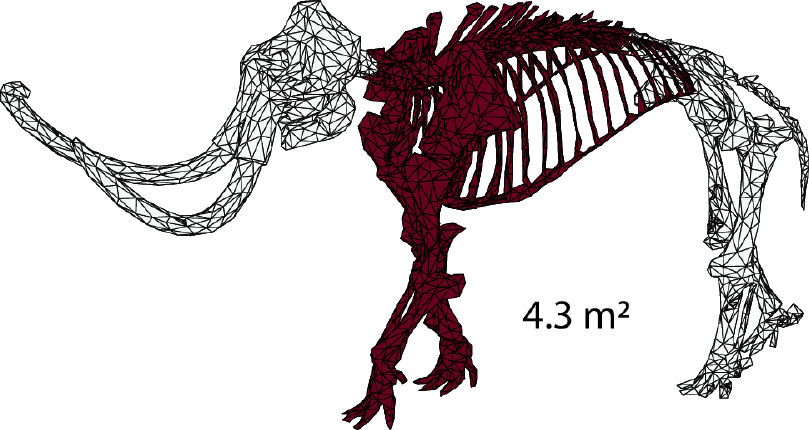} 
    \includegraphics[width=0.475\linewidth]{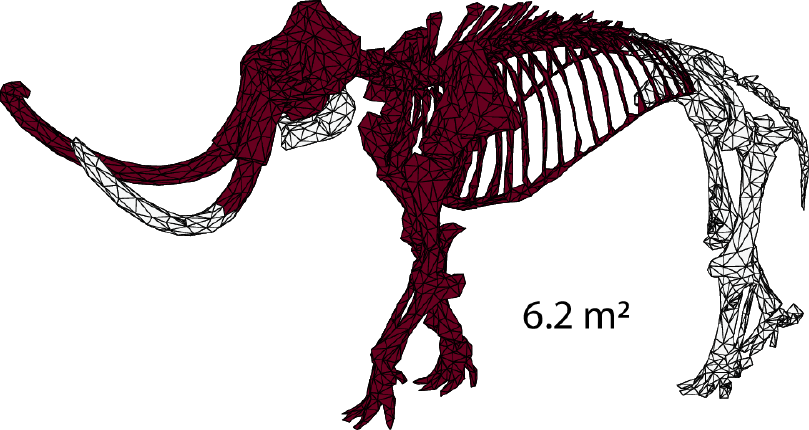} \\
    \includegraphics[width=0.475\linewidth]{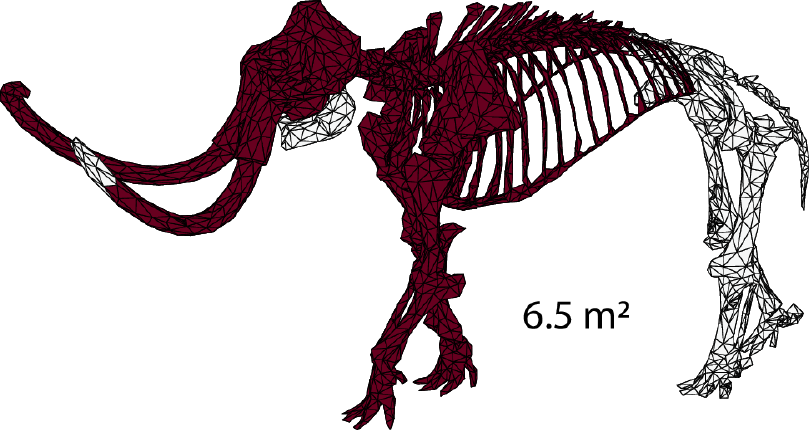} 
    \includegraphics[width=0.475\linewidth]{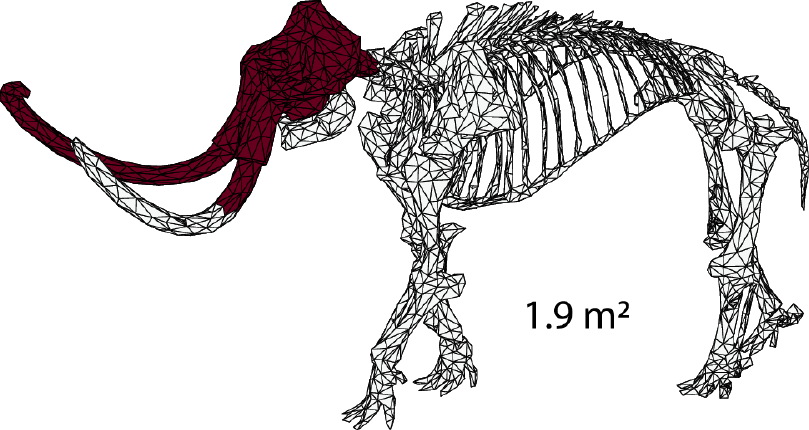} \\
    \includegraphics[width=0.475\linewidth]{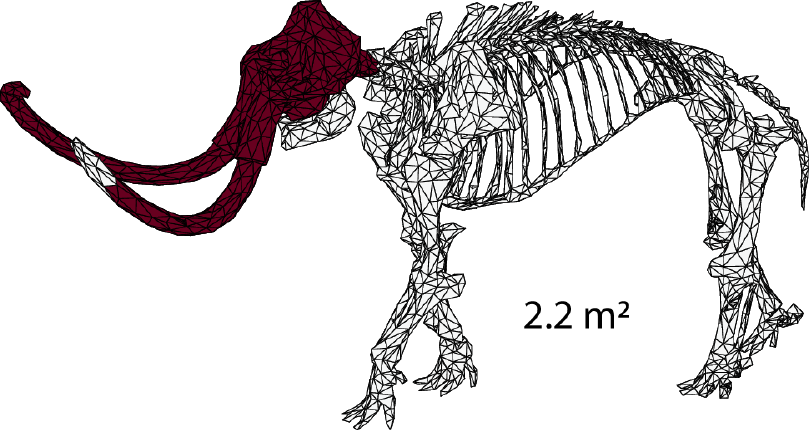} 
    \includegraphics[width=0.475\linewidth]{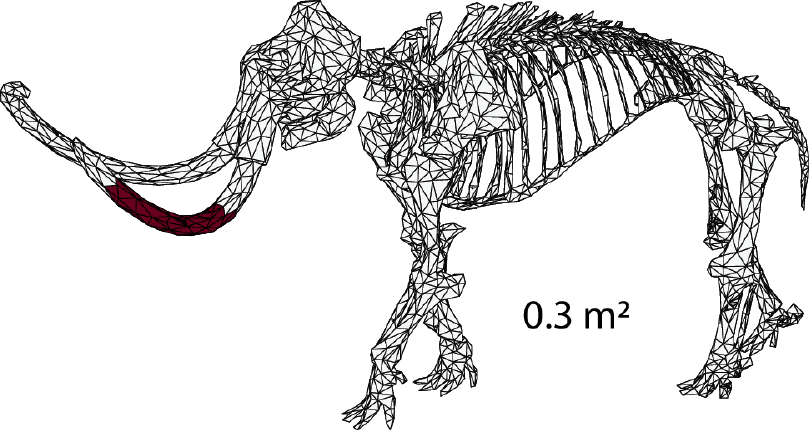} 
  \end{tabular}
  \caption{Minimum homology distances between cycles on a Smithsonian scanned model of a Mammoth skeleton }
  \label{fig:mammoth}
\end{figure*}

\begin{figure*}
  \centering
    \includegraphics[width=0.3\linewidth]{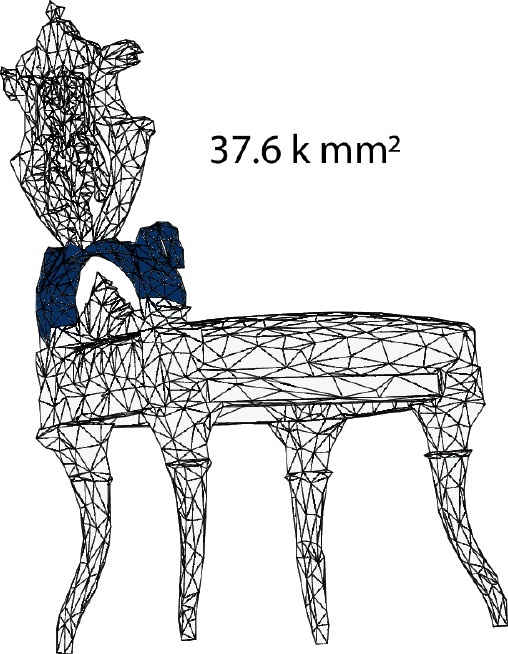} 
    \includegraphics[width=0.3\linewidth]{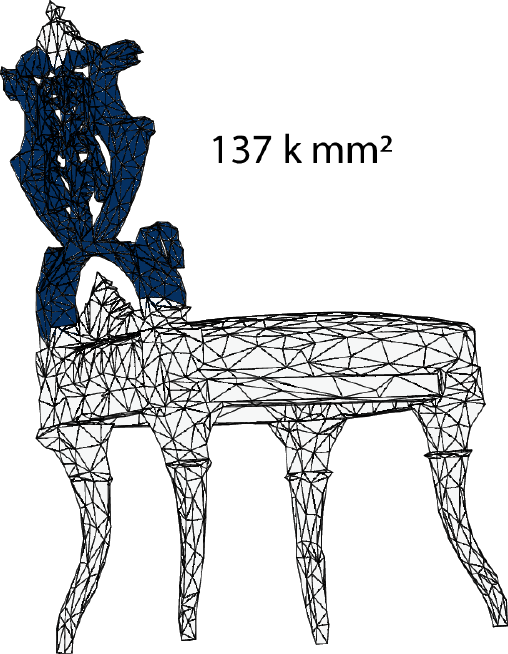} 
    \includegraphics[width=0.3\linewidth]{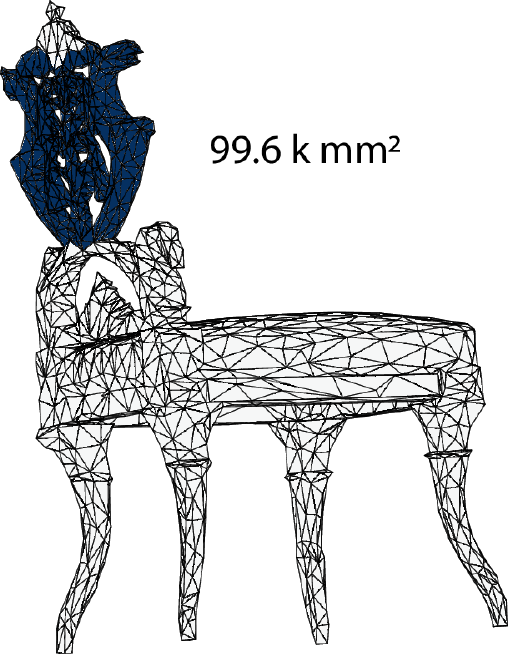} 
  \caption{Minimum homology distances between cycles on a Smithsonian
    scanned model of a renaissance chair. Notice how these bounding
    chains exemplify the absence of connecting homotopies described in
  Figure~\ref{fig:homotopyversushomologyonsurface}.}
  \label{fig:chair}
\end{figure*}

\begin{figure*}
  \centering
    \includegraphics[width=0.3\linewidth]{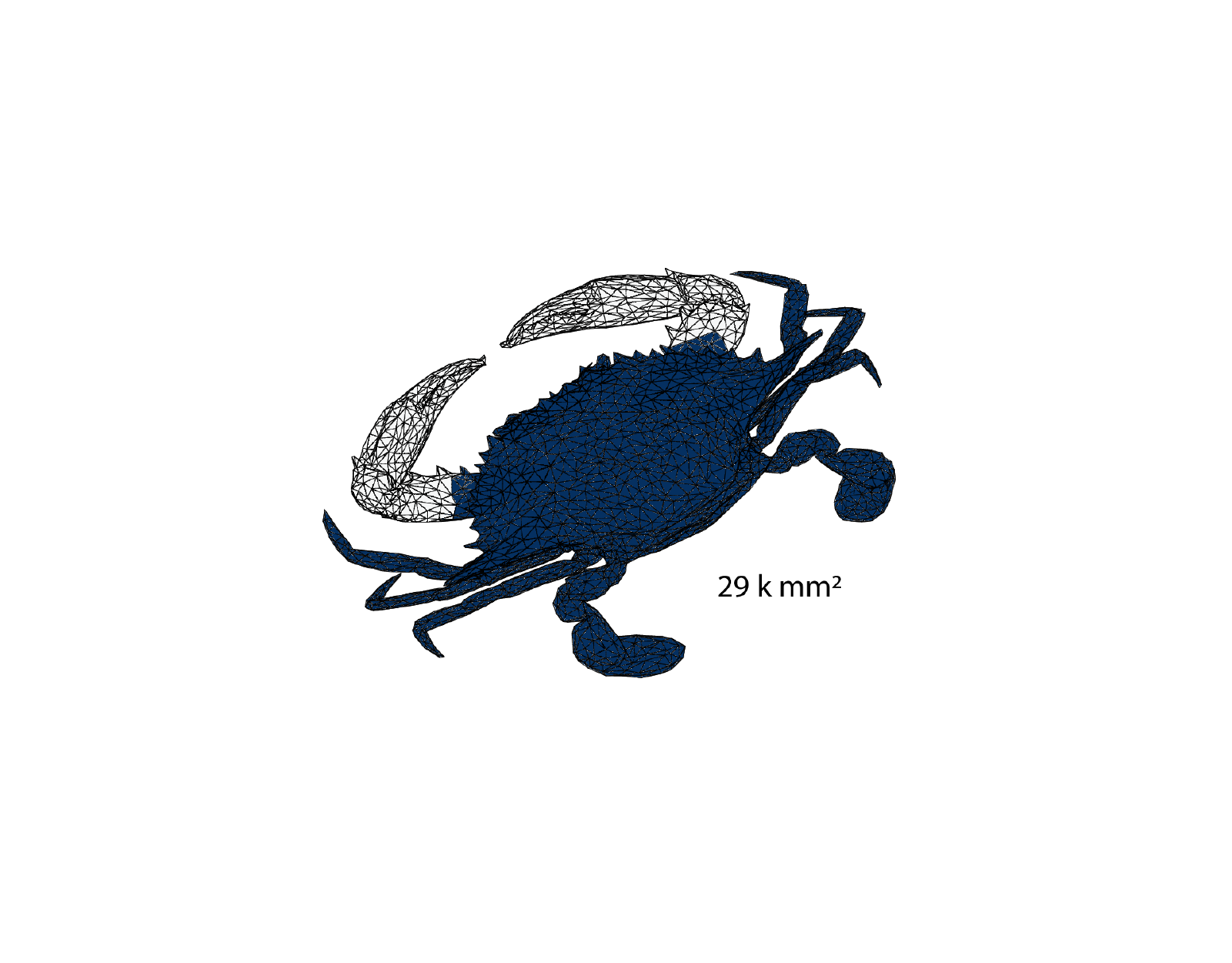} 
    \includegraphics[width=0.3\linewidth]{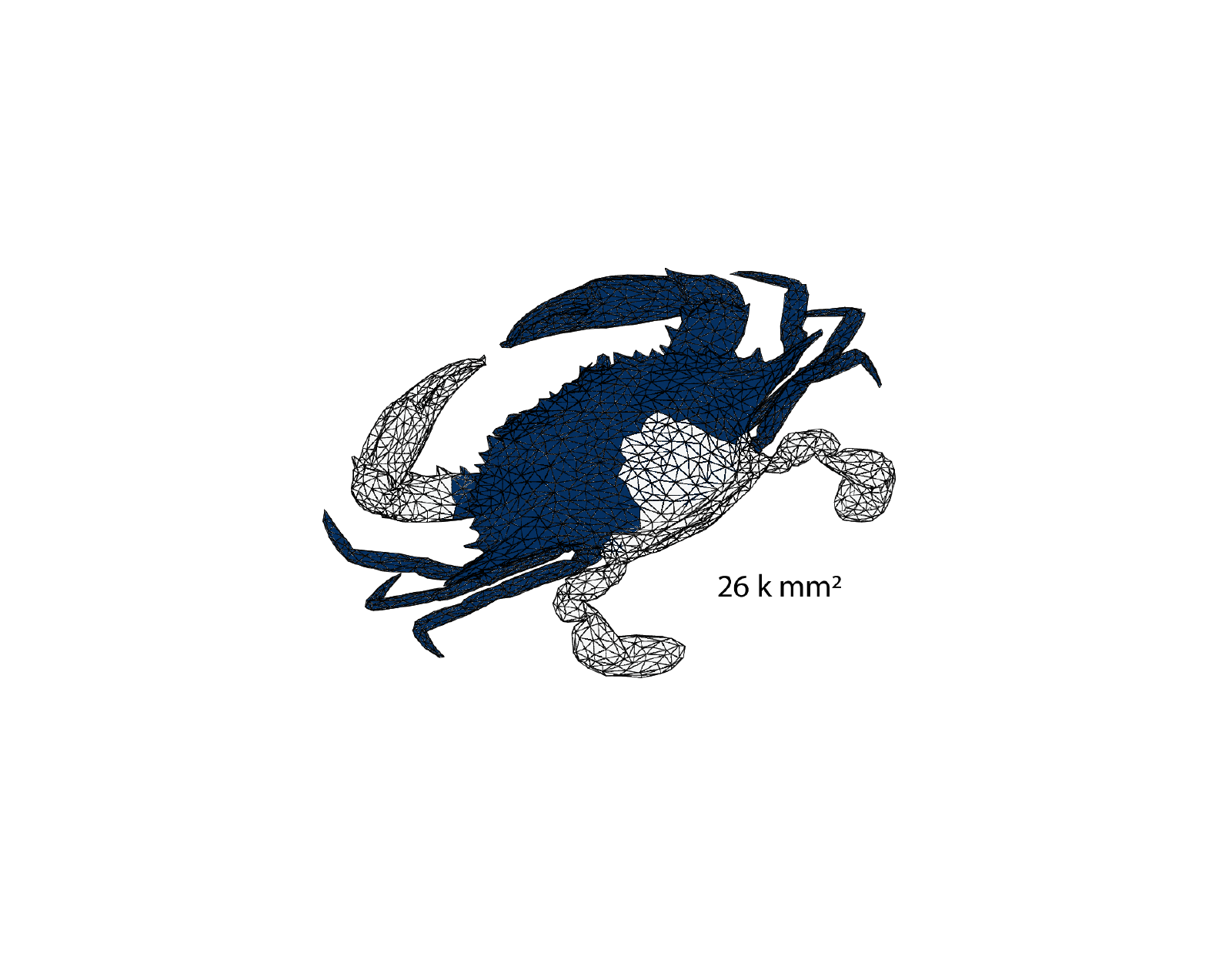} 
    \includegraphics[width=0.3\linewidth]{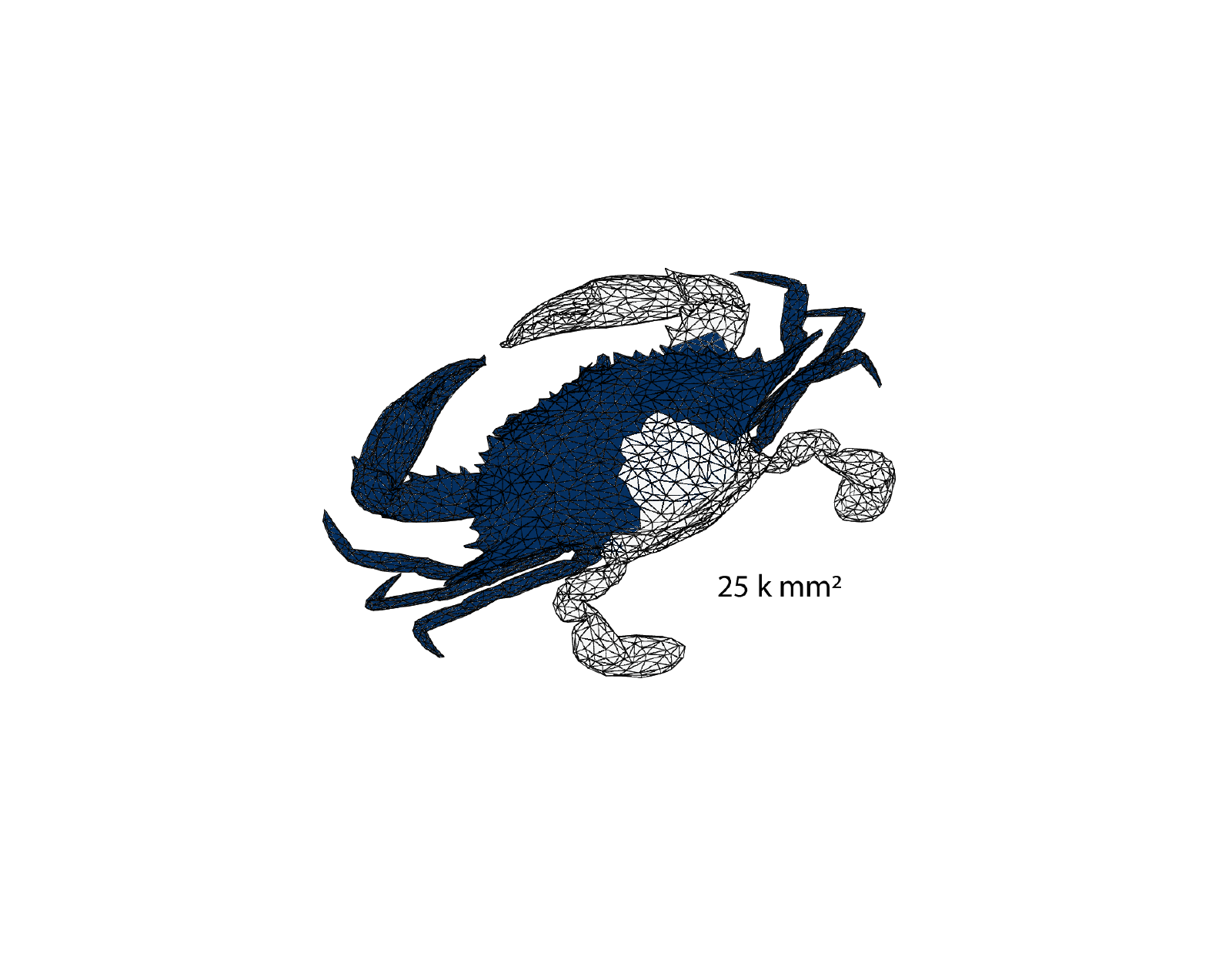} \\
    \includegraphics[width=0.3\linewidth]{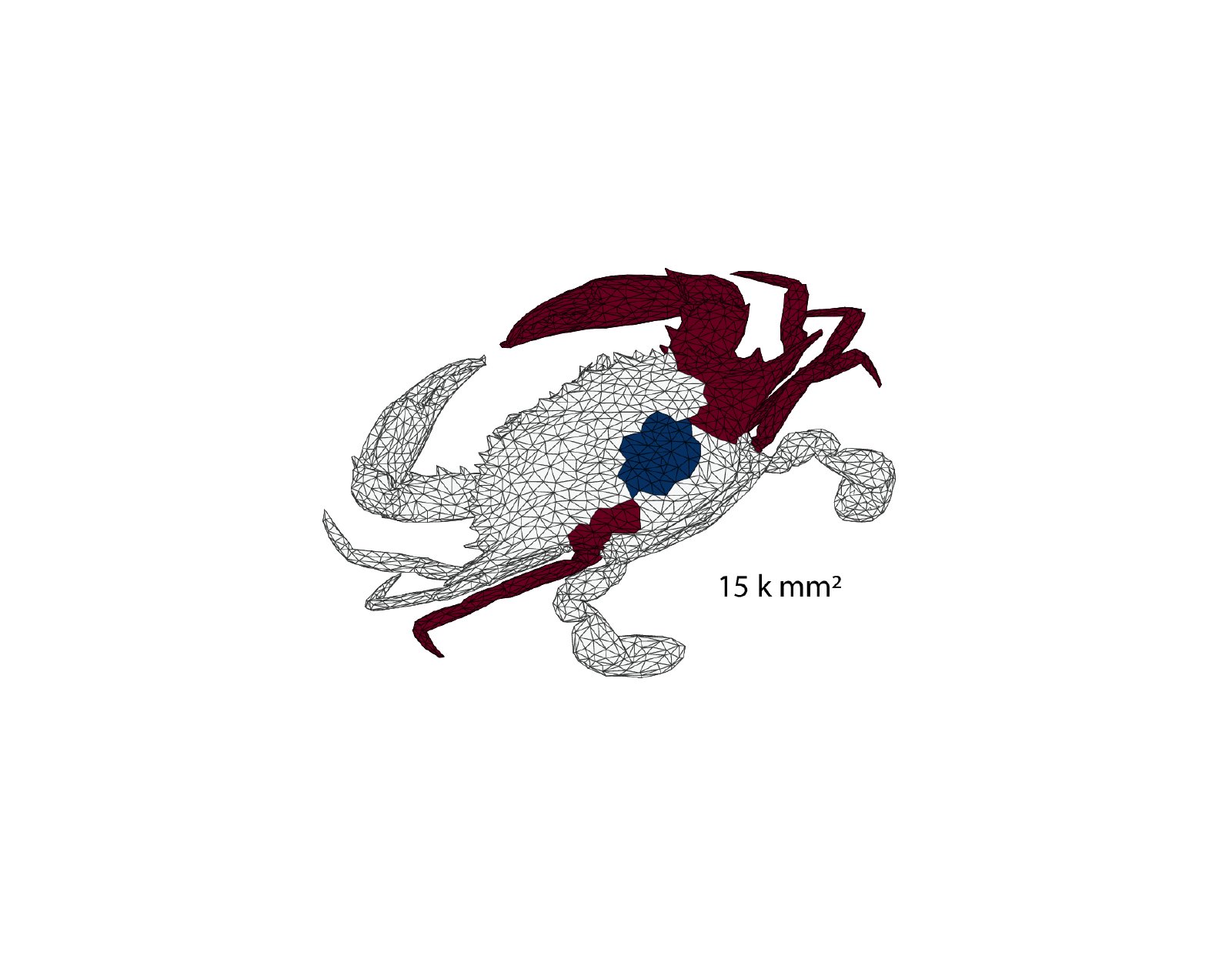} 
    \includegraphics[width=0.3\linewidth]{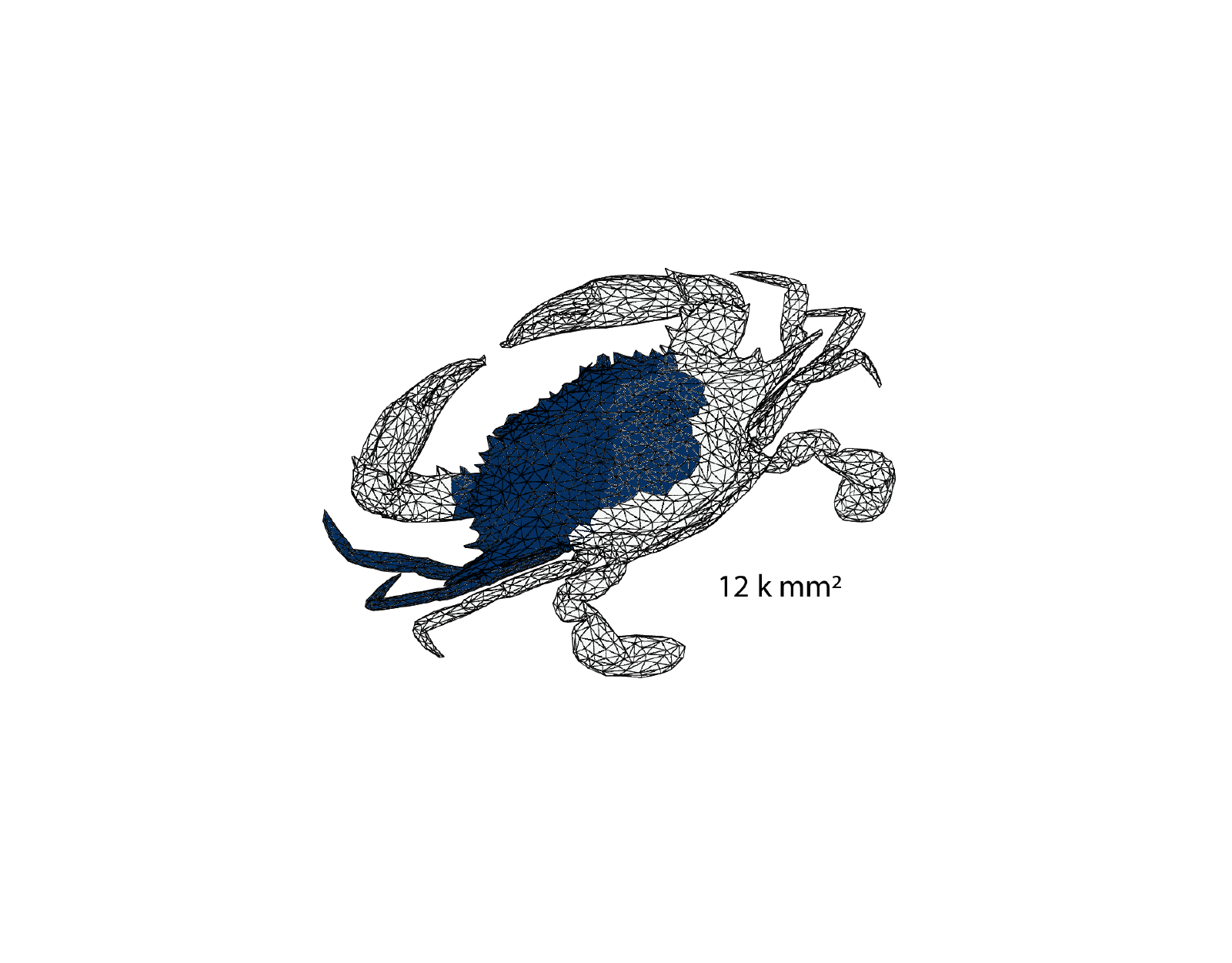} 
    \includegraphics[width=0.3\linewidth]{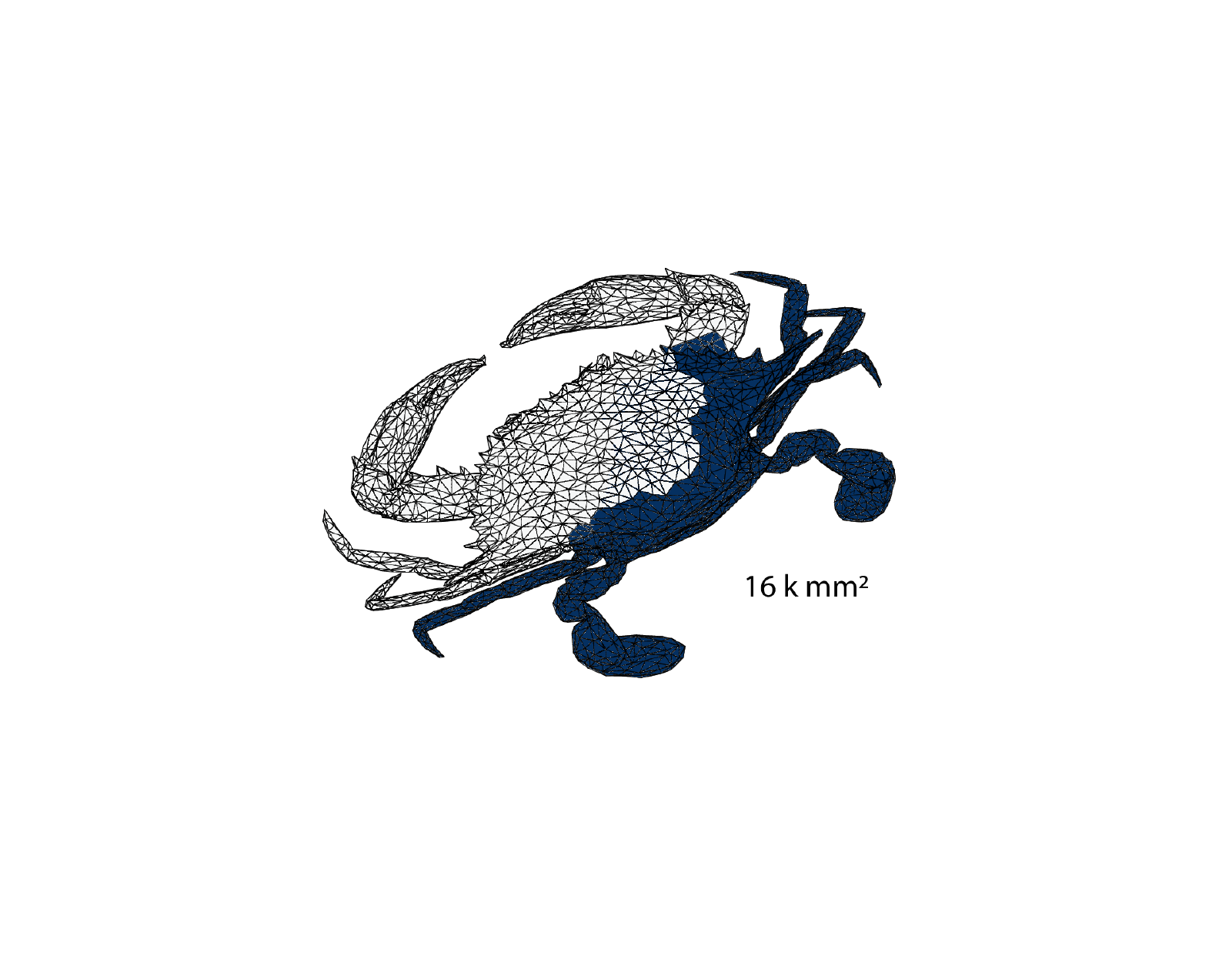} 
  \caption{Minimum homology distances between cycles on a Smithsonian scanned model of a crab. In the top row, cycles around each of the claw bases are connected to each other, and to a cycle around the shell. In the bottom row, another cycle around the shell is introduced and connected to each of the other cycles. Notice in particular how the intersecting cycles are handled at the bottom left: this is a case where the cycles do not form a simple cut graph, and the bounding chain is more involved than for our simpler examples.}
  \label{fig:crab}
\end{figure*}

\bibliographystyle{eg-alpha-doi}
\bibliography{curvesim}

\newcommand{\etalchar}[1]{$^{#1}$}
\begin{thebibliography}{\uppercase{AKPW14}}

\bibitem[AG95]{ag-cfdbt-95}
\textsc{Alt H., Godau M.}:
\newblock Computing the {Fr\'echet} distance between two polygonal curves.
\newblock \emph{International Journal of Computational Geometry and its
  Applications 5} (1995), 75--91.

\bibitem[AKPW14]{mapsurvey}
\textsc{Ahmed M., Karagiorgou S., Pfoser D., Wenk C.}:
\newblock A comparison and evaluation of map construction algorithms, 2014.
\newblock Arxiv.org pre-print.

\bibitem[Alt09]{alt2009}
\textsc{Alt H.}:
\newblock The computational geometry of comparing shapes.
\newblock In \emph{Efficient Algorithms}, vol.~5760 of \emph{Lecture Notes in
  Computer Science}. Springer Berlin / Heidelberg, 2009, pp.~235--248.

\bibitem[AS71]{asit}
\textsc{Atiyah M., Singer I.}:
\newblock The index of elliptic operators on compact.
\newblock \emph{Bull. Amer. Math. Soc. 93} (1971), 119--138.

\bibitem[BCP97]{magma}
\textsc{Bosma W., Cannon J., Playoust C.}:
\newblock The magma algebra system i: The user language.
\newblock \emph{Journal of Symbolic Computation 24}, 3 (1997), 235--265.

\bibitem[BKR13]{phat}
\textsc{Bauer U., Kerber M., Reininghaus J.}:
\newblock Clear and compress: Computing persistent homology in chunks.
\newblock \emph{arXiv preprint arXiv:1303.0477} (2013).

\bibitem[Bur]{burkardt}
\textsc{Burkardt J.}:
\newblock {PLY} files.
\newblock URL: \url{http://people.sc.fsu.edu/~jburkardt/data/ply/ply.html}.

\bibitem[CCE{\etalchar{*}}10]{CVE08}
\textsc{Chambers E.~W., {Colin de Verdi{\`e}re} {\'E}., Erickson J., Lazard S.,
  Lazarus F., Thite S.}:
\newblock Homotopic fr\'echet distance between curves or, walking your dog in
  the woods in polynomial time.
\newblock \emph{Computational Geometry 43}, 3 (2010), 295 -- 311.
\newblock Special Issue on 24th Annual Symposium on Computational Geometry
  (SoCG'08).
\newblock URL:
  \url{http://www.sciencedirect.com/science/article/pii/S0925772109000637},
  \href {http://dx.doi.org/10.1016/j.comgeo.2009.02.008}
  {\path{doi:10.1016/j.comgeo.2009.02.008}}.

\bibitem[CNR]{Meshlab}
\textsc{CNR V. C. L.~I.}:
\newblock Meshlab.
\newblock http://meshlab.sourceforge.net/.

\bibitem[CW10]{CW10}
\textsc{Cook A.~F., Wenk C.}:
\newblock Geodesic {Fr\'{e}chet} distance inside a simple polygon.
\newblock \emph{ACM Transactions on Algorithms 7} (2010).

\bibitem[CW13]{homotopyarea}
\textsc{Chambers E.~W., Wang Y.}:
\newblock Measuring similarity between curves on 2-manifolds via homotopy area.
\newblock In \emph{Proceedings of the Twenty-ninth Annual Symposium on
  Computational Geometry} (New York, NY, USA, 2013), SoCG '13, ACM,
  pp.~425--434.
\newblock URL: \url{http://doi.acm.org/10.1145/2462356.2462375}, \href
  {http://dx.doi.org/10.1145/2462356.2462375}
  {\path{doi:10.1145/2462356.2462375}}.

\bibitem[DHK11]{dhk2011}
\textsc{Dey T., Hirani A., Krishnamoorthy B.}:
\newblock Optimal homologous cycles, total unimodularity, and linear
  programming.
\newblock \emph{SIAM Journal on Computing 40}, 4 (2011), 1026--1044.
\newblock URL: \url{http://epubs.siam.org/doi/abs/10.1137/100800245}, \href
  {http://arxiv.org/abs/http://epubs.siam.org/doi/pdf/10.1137/100800245}
  {\path{arXiv:http://epubs.siam.org/doi/pdf/10.1137/100800245}}, \href
  {http://dx.doi.org/10.1137/100800245} {\path{doi:10.1137/100800245}}.

\bibitem[Don06]{donoho2006most}
\textsc{Donoho D.~L.}:
\newblock For most large underdetermined systems of linear equations the
  minimal 𝓁1-norm solution is also the sparsest solution.
\newblock \emph{Communications on pure and applied mathematics 59}, 6 (2006),
  797--829.

\bibitem[EG08]{hap}
\textsc{Ellis G., Galway N.}:
\newblock Homological algebra programming.
\newblock \emph{Computational group theory and the theory of groups 470}
  (2008), 63--74.

\bibitem[GAP14]{gap}
\textsc{The GAP~Group}:
\newblock \emph{{GAP -- Groups, Algorithms, and Programming, Version 4.7.4}},
  2014.
\newblock URL: \url{http://www.gap-system.org}.

\bibitem[GBLT]{shapely}
\textsc{Gillies S., Bierbaum A., Lautaportti K., Tonnhofer O.}:
\newblock Shapely.
\newblock URL: \url{http://toblerity.org/shapely}.

\bibitem[Hat02]{h-at-02}
\textsc{Hatcher A.}:
\newblock \emph{Algebraic Topology}.
\newblock Cambridge University Press, 2002.
\newblock URL: \url{http://www.math.cornell.edu/~hatcher/AT/ATpage.html}.

\bibitem[HSB66]{hrr}
\textsc{Hirzebruch F., Schwarzenberger R.~L., Borel A.}:
\newblock \emph{Topological methods in algebraic geometry}, vol.~232.
\newblock Springer, 1966.

\bibitem[JOP{\etalchar{*}}  ]{scipy}
\textsc{Jones E., Oliphant T., Peterson P., et~al.}:
\newblock {SciPy}: Open source scientific tools for {Python}, 2001--.
\newblock URL: \url{http://www.scipy.org/}.

\bibitem[KMM04]{chomp}
\textsc{Kaczynski T., Mischaikow K., Mrozek M.}:
\newblock \emph{Computational homology}, vol.~157.
\newblock Springer, 2004.

\bibitem[Law80]{lawson1980lectures}
\textsc{Lawson H.}:
\newblock \emph{Lectures on minimal submanifolds}, vol.~1 of \emph{Mathematics
  lecture series}.
\newblock Publish or Perish, 1980.
\newblock URL: \url{http://books.google.com/books?id=vlbvAAAAMAAJ}.

\bibitem[Mil07]{milnor}
\textsc{Milnor J.}:
\newblock A survey of cobordism theory.
\newblock \emph{Collected Papers of John Milnor: Differential topology 3}
  (2007), 291.

\bibitem[Mor12]{dionysus}
\textsc{Morozov D.}:
\newblock Dionysus, 2012.
\newblock URL: \url{http://www.mrzv.org/software/dionysus}.

\bibitem[Mun00]{m-t-00}
\textsc{Munkres J.~R.}:
\newblock \emph{Topology}, 2nd~ed.
\newblock Prentice-Hall, 2000.

\bibitem[MZ93]{mallat1993matching}
\textsc{Mallat S.~G., Zhang Z.}:
\newblock Matching pursuits with time-frequency dictionaries.
\newblock \emph{Signal Processing, IEEE Transactions on 41}, 12 (1993),
  3397--3415.

\bibitem[Nan12]{perseus}
\textsc{Nanda V.}:
\newblock Perseus: the persistent homology software, 2012.

\bibitem[PRK93]{342465}
\textsc{Pati Y., Rezaiifar R., Krishnaprasad P.~S.}:
\newblock Orthogonal matching pursuit: recursive function approximation with
  applications to wavelet decomposition.
\newblock In \emph{Signals, Systems and Computers, 1993. 1993 Conference Record
  of The Twenty-Seventh Asilomar Conference on} (Nov 1993), pp.~40--44 vol.1.
\newblock \href {http://dx.doi.org/10.1109/ACSSC.1993.342465}
  {\path{doi:10.1109/ACSSC.1993.342465}}.

\bibitem[PVG{\etalchar{*}}11]{scikit-learn}
\textsc{Pedregosa F., Varoquaux G., Gramfort A., Michel V., Thirion B., Grisel
  O., Blondel M., Prettenhofer P., Weiss R., Dubourg V., Vanderplas J., Passos
  A., Cournapeau D., Brucher M., Perrot M., Duchesnay E.}:
\newblock Scikit-learn: Machine learning in {P}ython.
\newblock \emph{Journal of Machine Learning Research 12} (2011), 2825--2830.

\bibitem[RS02]{kenzo}
\textsc{Rubio J., Sergeraert F.}:
\newblock Constructive algebraic topology.
\newblock \emph{Bulletin des Sciences Math{\'e}matiques 126}, 5 (2002),
  389--412.

\bibitem[RZE08]{rubinstein2008efficient}
\textsc{Rubinstein R., Zibulevsky M., Elad M.}:
\newblock Efficient implementation of the k-svd algorithm using batch
  orthogonal matching pursuit.
\newblock \emph{CS Technion} (2008), 40.

\bibitem[SHI96]{517122}
\textsc{Shum H.-Y., Hebert M., Ikeuchi K.}:
\newblock On 3d shape similarity.
\newblock In \emph{Computer Vision and Pattern Recognition, 1996. Proceedings
  CVPR '96, 1996 IEEE Computer Society Conference on} (Jun 1996), pp.~526--531.
\newblock \href {http://dx.doi.org/10.1109/CVPR.1996.517122}
  {\path{doi:10.1109/CVPR.1996.517122}}.

\bibitem[SJ05]{sage}
\textsc{Stein W., Joyner D.}:
\newblock Sage: System for algebra and geometry experimentation.
\newblock \emph{Communications in Computer Algebra (SIGSAM Bulletin)(July
  2005)} (2005).
\newblock URL: \url{http://sage.sourceforge.net}.

\bibitem[{Smi}a]{3dsi-chair}
\textsc{{Smithsonian X 3D}}:
\newblock Pergolesi chair, cooper-hewitt, national design museum.
\newblock URL: \url{https://3d.si.edu/downloads/64}.

\bibitem[{Smi}b]{3dsi-mammoth}
\textsc{{Smithsonian X 3D}}:
\newblock Woolly mammoth, national museum of natural history, usnm 23792.
\newblock URL: \url{https://3d.si.edu/downloads/55}.

\bibitem[TV04]{1314502}
\textsc{Tangelder J., Veltkamp R.}:
\newblock A survey of content based 3d shape retrieval methods.
\newblock In \emph{Shape Modeling Applications, 2004. Proceedings} (June 2004),
  pp.~145--156.
\newblock \href {http://dx.doi.org/10.1109/SMI.2004.1314502}
  {\path{doi:10.1109/SMI.2004.1314502}}.

\bibitem[TVJA12]{javaplex}
\textsc{Tausz A., Vejdemo-Johansson M., Adams H.}:
\newblock javaplex: a research platform for persistent homology.
\newblock \emph{Book of Abstracts Minisymposium on Publicly Available
  Geometric/Topological Software} (2012), 7.

\bibitem[VJ12]{pgap}
\textsc{Vejdemo-Johansson M.}:
\newblock Gap persistence--a computational topology package for gap.
\newblock \emph{Book of Abstracts Minisymposium on Publicly Available
  Geometric/Topological Software} (2012), 43.

\bibitem[Whi84]{white1984}
\textsc{White B.}:
\newblock Mappings that minimize area in their homotopy classes.
\newblock \emph{Journal of Differential Geometry 20}, 2 (1984), 433--446.
\newblock URL: \url{http://projecteuclid.org/euclid.jdg/1214439286}.

\end{thebibliography}

\end{document}